\documentclass[a4paper]{cas-sc}

\usepackage[authoryear,longnamesfirst]{natbib}
\usepackage{color}
\usepackage{tikz}
\usepackage[normalem]{ulem}

 \newtheorem{theorem}{Theorem}
  \newtheorem{proposition}{Proposition}
 \newtheorem{lemma}[theorem]{Lemma}
  \newtheorem{corollary}[theorem]{Corollary}
 \newdefinition{remark}{Remark}
  \newdefinition{example}{Example}
 \newproof{proof}{Proof}
 \usepackage{caption}
 
\def\tsc#1{\csdef{#1}{\textsc{\lowercase{#1}}\xspace}}
\tsc{WGM}
\tsc{QE}
\tsc{EP}
\tsc{PMS}
\tsc{BEC}
\tsc{DE}
\newcommand{\F}{\mathbb{F}}
\newcommand{\C}{\mathcal{C}}
\newcommand{\Mod}[1]{\ (\operatorname{mod}\ #1)}
\definecolor{darkred}{HTML}{BB0000}

\begin{document}
\let\WriteBookmarks\relax
\def\floatpagepagefraction{1}
\def\textpagefraction{.001}
\shorttitle{Infinite Families of Good Self-Orthogonal Quasi-Cyclic Codes }
\shortauthors{G. Terra Bastos, A. \'Alvarez, C. L. Williams}

\title [mode = title]{A Construction of Infinite Families of Good Self-Orthogonal Quasi-Cyclic Codes}                      

\tnotetext[1]{This document is the result of a research
   project funded by FAPEMIG RED-00133-21 and AMS-Simons Research Enhancement Grants for PUI Faculty.}


\author[1]{Gustavo {Terra Bastos}}[type=editor,
                        auid=000,bioid=1,
                        prefix=,
                        role=,
                        orcid=0000-0002-8263-5119]
\cormark[1]
\fnmark[1]
\ead{gtbastos@ufsj.edu.br}
\affiliation[1]{organization={Department of Mathematics and Statistics, Federal University of S\~ao Jo\~ao del-Rei},
                addressline={170 Frei Orlando Sq}, 
                city={S\~ao Jo\~ao del-Rei},
                postcode={36307-352}, 
                state={Minas Gerais},
                country={Brazil}}

\author[2]{Angelynn \'Alvarez}[type=editor,
                        auid=000,bioid=1,
                        prefix=,
                        role=,
                        orcid=0000-0003-0208-2831]
\ead{alvara44@erau.edu}
\affiliation[2]{organization={Department of Mathematics, Embry-Riddle Aeronautical University},
                addressline={3700 Willow Creek Rd.}, 
                city={Prescott},
                postcode={86301}, 
                state={Arizona},
                country={United States}}

\author[2]{Cameron L. Williams}[type=editor,
                        auid=000,bioid=1,
                        prefix=,
                        role=,
                        orcid=0000-0002-2025-7394]
\ead{willc187@erau.edu}

\cortext[cor1]{Corresponding author}


\begin{abstract}
Quasi-cyclic codes have been recently employed in the constructions of quantum error-correcting codes. In this paper, we propose a construction of infinite families of quasi-cyclic codes over $\F_q$ which are self-orthogonal with respect to the Euclidean and Hermitian inner products. In particular, their dimension and a lower bound for their minimum distance are computed using their constituent codes defined over field extensions of $\mathbb{F}_q$. We also show that the lower bound for the minimum distance satisfies the square-root-like lower bound and also show how dual-containing and self-dual quasi-cyclic codes can arise from our construction. Using the CSS construction, we show the existence of quantum error-correcting codes with good parameters.
\end{abstract}


\maketitle

\section{Introduction}

In the study of infinite families of codes, the growth rate of the minimum distance with respect to the code length is an important measure of performance. The \textit{square-root bound} is a useful parameter for measuring how good a classical error-correcting code is; namely, given $n$ and $d$ the code length and the minimum distance of a code $C$, we say $C$ is a \textit{good code} if $d\ge \sqrt{n}.$ This definition was extended to the \textit{square-root-like bound} when there is a positive constant $c$ such that $d\ge c\sqrt{n}$.  Codes satisfying square-root-like bounds are desirable as they exhibit nontrivial distance growth while maintaining strong algebraic structure. Such bounds guarantee that the error-correcting capability of the codes increases with the code length. Moreover, the square-root bound frequently arises as a natural benchmark for structured algebraic code families, where achieving linear distance growth is often difficult. The construction of infinite families of codes satisfying such bounds has attracted significant attention, and for instance, the various codes constructed in \cite{Chen2023FourIF, squarerootbound, 28fcd52cbe024304880942df5f3617ce, WLZZ25} satisfy the square-root-like bound.

Beyond minimum distance considerations, structural properties of codes such as self-orthogonality, dual-containment, and self-duality play important roles in coding theory. Self-orthogonal codes arise in several areas, including lattice theory \cite{W98}, row-self-orthogonal matrices \cite{Massey92, Massey98}, and linear complementary dual codes \cite{Massey98}. Self-dual codes also exhibit deep connections with group theory \cite{Hoehn00,JLLX13}, lattices \cite{SloaneChapter99}, and design theory \cite{Sole90}, and are known to be asymptotically good \cite{MACWILLIAMS1972153}. The construction of codes that simultaneously satisfy the square-root-like bound together with certain duality conditions are also of interest. For instance, the authors in~\cite{infiniteselfdualfamilies} provide infinite families of (Euclidean and Hermitian) self-dual cyclic codes satisfying the square-root bound, and in~\cite{goodsocycliccodes}, the authors explore constructions of self-orthogonal cyclic codes whose minimum distances are larger than the square-root bound. Moreover, classical codes with self-orthogonal or dual-containing properties are useful in the construction of quantum error-correcting codes through celebrated work in \cite{Calderbank-Shor1996, Steane1996}, known as the CSS construction. Because of these connections and applications, determining the existence and constructions of self-orthogonal, dual-containing, and self-dual classical error-correcting codes is an important problem in coding theory.

According to~\cite{infiniteselfdualfamilies}, there are very few constructions of infinite families of self-dual codes whose minimum distances meet the square-root lower bound, which are depicted in that paper. To the best of our knowledge, the authors of~\cite{infiniteselfdualfamilies} in 2025 were the first to present infinite families of non-binary self-dual cyclic codes achieving a square-root-type lower bound on their minimum distances. Additionally, \cite{Calderbank} provides a distinct construction of an infinite family of self-dual quasi-cyclic codes.

Quasi-cyclic (QC) are have been suitable for constructing structured families of codes with desirable distance and duality properties. QC codes are a natural generalization of cyclic codes and are well-studied in the classical coding theory literature. They possess a rich algebraic structure~\cite{seguin,cem,lingsole} and have strong connections with convolutional codes \cite{SolomonTilborg}. Moreover, there exist QC codes that are asymptotically good \cite{longcodes}.  Very recently, some constructions of self-dual, dual-containing, and self-orthogonal QC codes have been explored to develop effective quantum error-correcting codes. For example in~\cite{Abdukhalikov-Bag-Panario2023, GalindoHernando,benjwal}, the authors obtained quantum codes with good parameters by studying the generators of the dual of a QC code to determine the conditions where such codes are dual-containing, self-orthogonal, or self-dual. As these properties are intrinsically related to the CSS construction, QECCs with good parameters are obtained.  

In this paper, we propose a construction of infinite families of self-orthogonal, dual-containing, and self-dual QC codes using their constituent code representation described in \cite{lingsole}, and whose lower bound for their minimum distances are determined using their concatenated description described in \cite{jensen, OO24}. We then show that the codes in our construction yield minimum distances that satisfy a square-root-like lower bound. We also explore the existence of QECCs with good parameters arising from these constructions. The paper is organized as follows: Section~\ref{preliminaries} provides background results related to QC codes and their constituent codes, and also includes the foundations of constructing QECCs from classical codes. In Section~\ref{mainsection}, we present the main results of this paper, namely the use of constituent codes of QC codes to construct infinite families of QC codes that are self-orthogonal (dual-containing, or self-dual) with respect to the Euclidean inner product. In Section~\ref{section:sqrtbound}, we show that minimum distances of such codes satisfy the square-root-like lower bound. We then introduce the notion of Galois closed codes in Section~\ref{Galois-Closed-Section} to obtain QC codes that are self-orthogonal (dual-containing, or self-dual)  with respect to the Hermitian inner product. Lastly, in Section~\ref{QECC_constructions}, we show how some codes from our construction yield quantum error-correcting codes with good parameters using the CSS construction. All computations throughout this paper have been done using \textsc{Magma, SageMath}, and \textsc{Mathematica}.

\color{black}

\section{Preliminaries}\label{preliminaries}

 Let $\F_q$ be a finite field with $q$ elements, where $q = p^t$, for $p$ a positive prime and $t \in \mathbb{Z}^+$. An $[n,k,d]_q$-\textit{linear code }$C$ is a $k$-dimensional $\mathbb{F}_q$-vector subspace of $\mathbb{F}_q^{n}$, where the minimum weight of nonzero codewords is $d$. In this work, only linear codes are presented, so we will refer to them as \textit{codes} only. The relationship between the parameters $n,k$ and $d$ is the following Singleton bound
\begin{equation}\label{singleton}
d \leq n - k +1
\end{equation}
and the code is called a \textit{maximum distance separable} code, or an MDS code, if equality in \eqref{singleton} is achieved. A quintessential example of a maximum distance separable code is the Reed-Solomon code described in \cite{RS1958}. More generally, if we choose nonzero elements $v_1,...,v_n \in \F_q$ and distinct elements $\alpha_1,...,\alpha_n \in \F_q$, we define the Generalized Reed-Solomon (GRS) Codes as
\begin{equation}\label{GRS}
GRS_{n,k}(\mathbf{\alpha}, \mathbf{v}) = \{ (v_1f(\alpha_1), ...,v_nf(\alpha_n)) \mid f(x) \in \F_q [x]_k\}
\end{equation}
where $\textbf{v}=(v_1,...,v_n)$, $\textbf{$\alpha$}=(\alpha_1,...,\alpha_n)$, and $\F_q [x]_k$ is the set of polynomials in $\F_q [x]$ of degree less than $k$. This family of codes is another example of an MDS code that will be useful in Example~\ref{exinffamilysqroot}.

This paper focuses on quasi-cyclic codes, which are generalizations of cyclic codes. More formally, a \textit{quasi-cyclic (QC) code} of index $\ell$ is a code $\C \subset \mathbb{F}_{q}^{m\ell}$ where for all $\textbf{c} \in \C$, $T^\ell(\textbf{c}) \in \C$, where $T$ is the standard shift operator on $\mathbb{F}_{q}^{m\ell}$ and $\ell$ is the smallest positive integer with this property.  
If $\C$ is a QC code in $\mathbb{F}_{q}^{m\ell}$, then any codeword $\textbf{c} \in \C$ can be represented as an $m \times \ell$ array
\begin{equation}\label{codeword}
\textbf{c}=\begin{pmatrix}
c_{00}& c_{01}&...&c_{0,\ell-1} \\ 
c_{10}& c_{11}&...&c_{1,\ell-1} \\ 
\vdots&\vdots&\ddots&\vdots\\ 
c_{m-1,0}& c_{m-1,1}&...&c_{m-1,\ell-1} 
\end{pmatrix}
 \end{equation}
where the action of the operator $T^{\ell}$ is equivalent to a row shift in~\eqref{codeword}.


To describe the constituent code representation of a QC code in \cite{lingsole}, we consider the ring $R_{q,m} ^{\ell}:=\left(\mathbb{F}_{q}[x] / \langle x^m-1 \rangle\right)^{\ell}$, where $\operatorname{gcd}(m,q)=1$. To any codeword $\textbf{c}$ as in \eqref{codeword}, we associate the following element of $R_{q,m} ^{\ell}$:
 \begin{equation}\label{polynomial_rep}
\vec{\textbf{c}}(x) := \left( c_0(x), c_1(x),...,c_{\ell-1}(x)\right) \in R_{q,m} ^{\ell}
\end{equation}
where for each $0 \leq j \leq \ell-1$,
\begin{equation*}
c_{j}(x):=c_{0j}+c_{1j}x+c_{2j}x^2+\cdots+c_{m-1,j}x^{m-1} \in R_{q,m}.
\end{equation*}

\noindent By Lemma 3.1 in \cite{lingsole}, the map $\phi: \mathbb{F}_{q}^{m\ell} \to R_{q,m}^\ell$ defined by $\phi(\textbf{c})=\vec{\textbf{c}}(x)$ is an $R_{q,m}$-module isomorphism. Hence, any QC code over $\mathbb{F}_{q}$ of length $m\ell$ and index $\ell$ can be viewed as an $R_{q,m}$-submodule of $R_{q,m}^\ell$. 

Additionally, given a polynomial $h(x) = h_n x^n + h_{n-1} x^{n-1} + \cdots + h_0 \in \F_q[x]$, its \textit{reciprocal polynomial} $h^* (x)$ is given by $h^*(x) = x^n h(x^{-1})$. As seen in~\cite{lingsole}, we may decompose $x^m-1 \in \F_q[x]$ into irreducible factors
\begin{equation}\label{decomp}
x^m -1 = \delta \prod_{j=1} ^t g_j (x)g_j ^* (x) \prod_{i=1} ^s f_i (x),
\end{equation}

\noindent where $\delta$ is a non-zero element in $\F_q$, $f_i (x)$ is a self-reciprocal polynomial for each $1\leq i\leq s$, and $g_i^*(x)$ is the reciprocal polynomial of $g_i(x)$, where $g_i(x) \neq g_i^*(x)$, for each $1\leq j\leq t$. With this decomposition and via the Chinese remainder theorem (CRT), the ring $R_{q,m}$ may be written as
\begin{equation}\label{decomposition_of_R}
R_{q,m}=\frac{\mathbb{F}_{q}[x]}{\langle x^m-1 \rangle} \cong \Bigg(\bigoplus_{j=1}^t \bigg(\frac{\mathbb{F}_{q}[x]}{\langle g_j \rangle} \oplus \frac{\mathbb{F}_{q}[x]}{\langle g_j^* \rangle} \bigg)\Bigg) \oplus \bigg(\bigoplus_{i=1} ^s \frac{\mathbb{F}_{q}[x]}{\langle f_i\rangle} \bigg),
\end{equation}
%
where we have suppressed the explicit dependence on $x$ in the $g_j$ and $f_i$ to simplify notation. Let $\alpha$ be an $m$-th primitive root of the unity in an extension of $\F_q$ and let
\[ F_i:=\dfrac{\mathbb{F}_{q}[x]}{\langle f_i\rangle} \cong \mathbb{F}_{q} (\alpha^{u_i}), \quad G'_j:=\dfrac{\mathbb{F}_{q}[x]}{\langle g_j\rangle} \cong \mathbb{F}_{q} (\alpha^{v_j}), \quad \text{and} \quad G''_j:=\dfrac{\mathbb{F}_{q}[x]}{\langle g^{*}_j\rangle} \cong \mathbb{F}_{q} (\alpha^{v^{*} _j}).\] 
for each $1 \leq i \leq s$ and $1 \leq j\leq t$. Hence,
\begin{align}\label{decomposition_of_R^ell}
R_{q,m}^\ell & \cong \Bigg(\bigoplus_{j=1}^t\left(\frac{\F_{q} [x]}{\langle g_j \rangle} \oplus \frac{\F_{q} [x]}{\langle g_j^* \rangle} \right)^{\ell}\Bigg) \oplus \left( \bigoplus_{i=1}^s \left(\frac{\F_{q} [x]}{\langle f_i \rangle} \right)^{\ell}  \right) \\
& \cong  \Bigg(\bigoplus_{j=1}^t (G'_j \oplus G''_j )^\ell\Bigg) \oplus \bigg(\bigoplus_{i=1} ^s F_{i}^\ell \bigg) .
\end{align}
Define codes $C_i \subseteq F_i^\ell$, $C'_j \subseteq G_j^ {' \ell}$, and $C''_j \subseteq G_j ^{'' \ell}$ for $1\leq i \leq s$ and $1 \leq j \leq t$. We call these codes the \textit{constituent codes} of the QC code $\mathcal{C}$ which has the following CRT decomposition:
\begin{equation}\label{CRTdecompcode}
\C \cong \Bigg(\bigoplus_{j=1}^t (C_j ' \oplus C_j '')\Bigg) \oplus \bigg(\bigoplus_{i=1} ^s C_i \bigg).
\end{equation}

The following lemma gives an easy way to compute the dimension of a QC code $\C$ from its constituent codes. We slightly modified the notation of the following result from \cite{BastosAlvarezFloresSalerno} to align with the notation in this paper.

\begin{lemma}\cite[Lemma 2]{BastosAlvarezFloresSalerno}\label{dimensionlemma}
Let $\C\subset \mathbb{F}_q ^{m\ell}$ be an $[m\ell,k]$ QC code and let $C_{\mu}\subseteq  \left(\frac{\F_q [x]}{\langle b_j\rangle}\right)^{\ell}$ be its respective constituent $[\ell, k_{\mu}]$ codes. If the nontrivial constituent codes of $\C$ are given by $C_1, \ldots, C_h$, then
$\displaystyle k = \sum_{\mu=1}^h k_\mu \deg(b_\mu). $
\end{lemma}
Moreover, if the nonzero constituent codes of $\C$ are denoted by $C_1, C_2,...,C_h$, we assume from now on that their minimum distances satisfy the following \textit{descending-chain condition}
\begin{equation}\label{chaincond}
d(C_1)\geq d(C_2) \geq \cdots \geq d(C_h),
\end{equation}
which will be utilized when computing a lower bound for the minimum distance of $\C$.

\color{black}

\subsection{The Cyclic Code Associated to a QC Code and a Lower Bound for its Minimum Distance}

Let $s \in \mathbb{Z}_m$ and let $C_s = \{s, sq, sq^2, ... ,sq^{r-1}\}\subset \mathbb{Z}_m$ be the $q$-cyclotomic coset of $s$, where $r$ is the smallest positive integer so that $sq^r\equiv s \mod m$. It is important to note that the $q$-cyclotomic cosets split $\mathbb{Z}_m.$ This classic concept is seen in constructions of cyclic codes, as well as in the trace representation of QC codes in ~\cite{lingsole}.

Let $C_1 \subset \F_q ^{\ell}(\alpha^{u_{a_1}}),...,C_h \subset \F_q ^{\ell} (\alpha^{u_{a_h}})$ be the nonzero constituent codes of an $\ell$-QC code $\C\subset \F_q ^{m\ell}$. By~\cite[Lemma 4.1]{cem}, the trace-formula described in \cite[Theorem 5.1]{lingsole} may be taken from an specific finite extension field $\mathbb{F}_{q^w}$, where $\mathbb{F}_{q} (\alpha^{u_{a_i}}) \subset \mathbb{F}_{q^w}$, for all $1\leq i\leq h$, given the existence of elements $w_i \in \mathbb{F}_{q} (\alpha^{u_{a_i}})$ so that $\operatorname{Tr}_{ \mathbb{F}_{q^w} / \mathbb{F}_{q} (\alpha^{u_{a_i}})} (w_i) =1$. In this case, when taking $\C$ as a collection of arrays, the columns of any codeword $\textbf{c} \in \C$ are codewords of an $m$-length cyclic code $D:=D_{1,2,..,h}$, whose dual's basic zero set is
$$BZ(D^{\perp})=\{ \alpha^{-u_{a_{i}}} \mid i=1,2,...,h\}.$$ 
Recall that if $D$ is a cyclic code of length $m$ over $\mathbb{F}_q$ and $\alpha$ is a primitive $m$-th root of unity, the basic zero set of $D^\perp$ consists of those powers $\alpha^{-u}$ at which every codeword polynomial of $D^\perp$ vanishes. Equivalently, it is the set of roots of the generator polynomial of $D^\perp$, described via representatives of their $q$-cyclotomic cosets modulo $m$.
Just as in~\cite{BastosAlvarezFloresSalerno}, we call $D$ the \emph{cyclic code associated to $\C$}. Let $I=\{i_1,i_2,...,i_t\}$ be a nonempty subset of $\{1,2,...,h\}$ where $$1 \leq i_1 < i_2 < \cdots < i_t \leq h.$$ Define $D_{I}:=D_{i_1,...,i_t} \subset D$ to be the cyclic subcode of $D$ whose dual's basic zero set is
\begin{equation*}\label{bz}
BZ(D_{I}^{\perp})=\{ \alpha^{-u_{a_i}} \mid i \in I\}
\end{equation*}
so that $D_I=D_{i_1} \oplus D_{i_2} \oplus \cdots \oplus D_{i_t}$ and $d(D_I) \geq d(D)$ for each $I \subset \{1,2,...,h\}$.

The following result uses Jensen’s concatenated description of QC codes \cite[Theorem 2.1]{OO24}, and their associated cyclic codes and subcodes to establish a lower bound on the minimum distance of a QC code. This lower bound will be referred to \textit{Jensen's bound} where relevant details can be be found in \cite{jensen, OO24}.





\begin{theorem}\cite[Theorem 4]{jensen}\label{jensenbound}
Given a QC code $\C$ of length $m\ell$ and index $\ell$ over $\F_q$ with the concatenated structure $\displaystyle{\C \cong \bigoplus_{i=1} ^h D_i \square C_i },$ and $d(C_1) \geq d(C_2) \geq ...\geq d(C_h)$, we have
\begin{equation}\label{jensen_lowerbound_equation}
d(\C) \geq d_{\text{J}}(\C):=\min\{N_h , N_{h-1,h}, \ldots, N_{1,2,\ldots,h}\},
\end{equation}
where
\begin{align*}
    N_h (\C) & = d(C_h) d(D_h),\\
    N_{h-1,h} (\C) & = d(C_{h-1})d(D_{h-1,h}), \\
    & \vdots \\
    N_{1,2,\ldots, h}(\C) &= d(C_1) d(D_{1,2,\ldots,h}).
    \end{align*}
\end{theorem}

Using the CRT decomposition~\eqref{CRTdecompcode} and Jensen's concatenated description, we know that a QC code $\C$ can be decomposed as the following direct sum of concatenations

\begin{equation}\label{CRT_Jensen_Decomposition}
\C \cong \left(\bigoplus_{j=1}^t \left(D_j'\square C_j' \oplus D_j''\square C_j'' \right) \right) \oplus \left( \bigoplus_{i=1}^
s D_i \square C_i \right).
\end{equation}
We will apply Theorem~\ref{jensenbound} to the QC code representation in \eqref{CRT_Jensen_Decomposition} to determine a lower bound for the minimum distance for the codes in our QC family in Section~\ref{mainsection}.

\subsection{The Duals of Quasi-Cyclic Codes}

Recently, QC codes have been used in constructions of QECCs where orthogonality conditions on QC codes play a central role \cite{Abdukhalikov-Bag-Panario2023,GalindoHernando,benjwal} . To study the duality of codes in $\mathbb{F}_q^{m\ell}$, we define the following inner products: Let $\textbf{a}, \textbf{b} \in \mathbb{F}_{q}^{m\ell}$ where $\textbf{a}=(a_{ij})$ and $\textbf{b}=(b_{ij})$ where $0 \leq i\leq m-1$ and $0\leq j\le \ell-1$.
The \textit{Euclidean inner product} on $\mathbb{F}_q^{m\ell}$ is
\begin{equation*}
\mathbf{a}\cdot\mathbf{b} = \langle \textbf{a}, \textbf{b} \rangle_{E}= \sum_{i=0}^{m-1} \sum_{j=0}^{\ell-1} a_{ij}b_{ij},
\end{equation*}

\noindent and the \textit{Hermitian inner product} on $\mathbb{F}_{q^t}^{m\ell}$, for even $t$ is
\begin{equation*}
\langle \textbf{a}, \textbf{b} \rangle_{H}= \sum_{i=0}^{m-1} \sum_{j=0}^{\ell-1} a_{ij} b_{ij} ^{\sqrt{q^t}}.
\end{equation*}

\noindent In particular, when $\textbf{a},\textbf{b}\in \F_{q^2} ^{m\ell}$, then $\displaystyle \langle \textbf{a}, \textbf{b} \rangle_{H}= \sum_{i=0} ^{m-1} \sum_{j=0}^{\ell-1} a_{ij} b_{ij} ^{q}$. On $R_{q,m}^\ell$, the \textit{Hermitian inner product} is as follows: for $\textbf{x}=(x_0,x_1,...,x_{\ell-1})$, and $\textbf{y}=(y_0, y_1,...,y_{\ell-1})$,
\begin{equation*}\label{hermitian__IP_R}
\langle \textbf{x}, \textbf{y} \rangle_{H} = \sum_{j=0}^{\ell-1} x_j\overline{y_j},
\end{equation*}
 
\noindent where $\bar{\;}$ denotes the ``conjugation map" on $R_{q,m}$ that sends $x$ to $x^{-1}=x^{m-1}$ and is extended $\F_{q}$-linearly.
%
%
\color{black}
The following proposition is useful for determining when orthogonality in $R_{q,m} ^\ell$ implies orthogonality in $\mathbb{F}_{q}^{m\ell}$.
\begin{proposition}\cite[Proposition 3.2]{lingsole}
Let $\mathbf{a}, \mathbf{b}\in \F_q ^{m\ell}$. Then $(T^{k\ell}(\mathbf{a}))\cdot \mathbf{b}=0$ for all $0\leq k\leq m-1$ if and only if $\langle \phi(\mathbf{a}),\phi(\mathbf{b}) \rangle_H = 0$, where $\phi$ is the aforementioned $R_{q,m}$-module isomorphism.
\end{proposition}
%
%
Let $\C$ be a QC code in $\mathbb{F}_{q}^{m\ell}$ and define its \textit{Euclidean dual} $\C^{\perp_E}$ by
\begin{equation*}
\C^{\perp_E} := \{ \textbf{x} \in \mathbb{F}_{q}^{m\ell} \mid \langle \textbf{x}, \textbf{c} \rangle_E =0, \forall \; \textbf{c} \in \C \}.
\end{equation*}

\noindent Similarly, the \textit{Hermitial dual}, $\C^{\perp_H}$ of $\C$ is given by
\begin{equation*}
\C^{\perp_H} := \{ \textbf{x} \in \mathbb{F}_{q^t}^{m\ell} \mid \langle \textbf{x}, \textbf{c} \rangle_H =0, \forall \; \textbf{c} \in \C \}.
\end{equation*}

\noindent We say that $\C$ is \textit{self-orthogonal} with respect to the Euclidean inner product (or ESO) if $\C \subseteq \C^{\perp_{E}}$, \textit{dual-containing} with respect to the Euclidean inner product (or EDC) if $\C^{\perp_{E}} \subseteq \C$, and \textit{self-dual}  with respect to the Euclidean inner product (or ESD) if $\C^{\perp_{E}} = \C$. Analogously, $\C$ is \textit{self-orthogonal} (HSO), \textit{dual-containing} (HDC), and \textit{self-dual} (HSD) with respect to the Hermitian inner product if $\C \subseteq \C^{\perp_{H}}, \C^{\perp_{H}} \subseteq \C$, and $\C=\C^{\perp_{H}}$, respectively.

The following result, which we slightly modified to better align with the notation in this paper, allows us to characterize the duality of QC codes using constituent codes and will be used throughout this paper to determine the duality of the QC codes in our code construction in Section~\ref{mainsection}.

\begin{theorem}\cite[Theorem 4.2]{lingsole}\label{lingsole} An $\ell$-quasi-cyclic code $\C$ of length $m\ell$ over $\mathbb{F}_{q}$ is self-dual with respect to the Euclidean inner product if and only if
\begin{equation}\label{lingsole-decomposition1}
\C\cong \left( \bigoplus_{j=1}^{t} ( C_j' \oplus (C'_{j})^{\perp_E})\right) \oplus \left( \bigoplus_{i=1}^s C_i \right) 
\end{equation}
where for $1 \leq i \leq s$, $C_i$ is a self-dual code over $F_i$ of length $\ell$ (with respect to the Hermitian inner product) and, for $1 \leq j \leq t$, $C'_j$ is a code of length $\ell$ over $G'_j$ and $C_j ^{''} = (C'_j)^{\perp_E}$ is its dual with respect to the Euclidean inner product over $G_j^{''}$.
\end{theorem}

\begin{remark}
\label{euclidean_hermitian_remark}
\end{remark}
\begin{enumerate}
\item[(i)] As observed in~\cite[page 136]{HKS}, the irreducible factors $f_i$ which are self-reciprocal imply that the cardinality $q_i$ of each field extension $F_i$ is an even power of $q$, for all $i$. However, there are two exceptions: The first one, for all $m$ and $q$, is the field extension coming from the irreducible factor $x-1$ of $x^m-1$. The other exception occurs when $q$ is odd and $m$ is even in which $x+1$ is another self-reciprocal factor of $x^m-1$. In both exceptions, $q_i=q$, so we can equip these fields with the Euclidean inner product. Otherwise, we equip $F_i$ with the Hermitian inner product.
\item[(ii)] When the constituent codes $C_i$ over finite extension fields coming from self-reciprocal factors are self-dual, self-orthogonal, or dual-containing (under the Hermitian or Euclidean inner products), for all $1 \leq i \leq s$, then the respective $\ell$-QC code $\C$ is self-dual, self-orthogonal, or dual-containing, respectively. Therefore, the duality of $\C$ is managed by the duality of its constituent codes over finite extension fields obtained from self-reciprocal polynomials.
 \end{enumerate}
%
\color{black}
The next result extends~\cite[Theorem 4.2]{lingsole} to the dual of a QC code.
\begin{proposition}\cite[Proposition 7.3.5]{HKS}\label{fulldecomp}
Let $\C$ be a QC code with CRT decomposition as~\eqref{CRTdecompcode}. Then its Euclidean dual $\C^{\perp_E}$ is of the form
\begin{equation*}\label{lingsole-decomposition2}
\C^{\perp_E} \cong \left( \bigoplus_{j=1}^{t} ( C_j' \oplus C_{j}^{' \perp_E})\right) \oplus \left( \bigoplus_{i=1}^s C_i ^{\perp_H}\right).
\end{equation*}
\end{proposition}
Since the duality of $\C$ can be determined from its constituent codes, we will use the following results in \cite{Li} for constructing self-orthogonal codes.

%
\begin{theorem}\cite[Theorem 1]{Li}\label{constbasedonprevioussocode}
\begin{enumerate}
\item[(i)] Suppose $C_1$ and $C_2$ are $[n_1, k , d_1]$ and $[n_2, k ,d_2]$ self-orthogonal codes over $\F_q$, respectively. If $C_1$ and $C_2$ have generator matrices $G_1$ and $G_2$, respectively, then $[G_1 | G_2]$ generates an $[n_1 +n_2 , k , d_1 + d_2]$ self-orthogonal code over $\F_q$.
\item[(ii)] Suppose $C_1$ and $C_2$ are $[n_1 , k , d_1]$ and $[n_2 , k-1, d_2]$ self-orthogonal codes over $\F_q$, respectively. If $C_1$ contains a codeword of weight at least $d_1 +d_2$, then there exists an $[n_1 + n_2, k, d_1 + d_2]$ self-orthogonal code over $\F_q$.
\end{enumerate}
\end{theorem}

In the special case when $C_1 = C_2$, the code with generator matrix $[G_1 | G_1]$ will be denoted as $2C_1$, and such an argument may be naturally extended to $mC_1$, for any positive integer $m$. In this work, for any code $C$, we will refer to the code $mC$ as \textit{$m$-copies of $C$} which will play an important role in the next section when determining the dimension of the QC codes in our infinite family.

\subsection{Constructions of Quantum Error-Correcting Codes from Classical Ones}\label{quantumconstructions}

From \cite{Ketkar-Klappenecker-Kumar-Sarvepalli2006}, the use of classical codes has became one of the standard methods for constructing quantum error-correcting codes. A \textit{quantum error-correcting code} (QECC) $Q$
is a $K$-dimensional subspace of $(\mathbb{C}^q)^{\otimes n}$. If $Q$ has minimum distance $d$, then we say that $Q$ is an $[[n, K, d]]_q$ code. If $K=q^k$, we write $[[n, k, d]]_q$. The length $n$, the dimension $K$, and minimum distance $d$ are the \textit{parameters} of $Q$. 

This paper focuses on the existence of quantum stabilizer codes, where a \textit{stabilizer (quantum) code} $Q\neq \{\textbf{0}\}$ is the common eigenspace of a commutative subgroup of the error group generated by a nice basis of $(\mathbb{C}^q)^{\otimes n}$. The code $Q$ is said to be \textit{pure to d}, or simply \textit{pure}, if and only if its stabilizer group does not contain non-scalar matrices of weight less than its minimum distance $d$. (See Definitions 3.5.7 and 3.5.8 in \cite{LaGuardia2021}.) Moreover, a code $Q$ is said to be \textit{impure} if and only if there are non-identity stabilizer elements of weight less than the minimum distance.
%
%
 In 1996, the following construction was introduced by Calderbank and Shor \cite{Calderbank-Shor1996}, and Steane \cite{Steane1996}, and is notably the most direct link between classical and quantum coding theory. 

\begin{theorem}[CSS construction \cite{Ketkar-Klappenecker-Kumar-Sarvepalli2006}]\label{css-construction} Let $C_1$ and $C_2$ denoted two classical codes with parameters $[n, k_1, d_1]_q$ and $[n, k_2, d_2]_q$ such that $C_2^\perp \subseteq C_{1}$. Then there exists $[[n, k_1+k_2-n, d]]_q$ stabilizer code with minimum distance $d=\min\{ \operatorname{wt}(\mathbf{c}) \mid \mathbf{c} \in (C_1 \backslash C_{2}^\perp) \cup (C_2 \backslash C_{1}^\perp) \}$ that is pure to $\min\{d_1, d_2\}$.
\end{theorem}
%
%
In the special case when $C_2=C_1$, we have that $C_1$ is a \textit{dual-containing} code. This yields the following result.
\begin{corollary}\cite{Aly-Klappenecker-Sarvepalli2007,Ketkar-Klappenecker-Kumar-Sarvepalli2006}\label{quantumcodefromdualcontaining}
Let $C$ be a classical $[n,k,d]$ code over $\mathbb{F}_q$. If $C$ is dual-containing, then there exists a quantum stabilizer code $Q$ with parameters $[[n, 2k-n, \geq d]]_q$ that is pure to $d$. Moreover, if the minimum distance of $C^\perp$ exceeds $d$, then the quantum code $Q$ has minimum distance $d$.
\end{corollary}

In recent years, researchers in coding theory have worked diligently to construct good (e.g., quantum MDS) quantum codes from classical Euclidean, Hermitian, or symplectic self-orthogonal, dual-containing, or self-dual codes \cite{Abdukhalikov-Bag-Panario2023, Ball21, Cao24, DinhLeNguyenTansuchat21,  FangFu19, GalindoHernando, galindo, GaoSunYanZhao25, JinXing13, LiuLiu21, ShiYueChang18, TianZhu19, WanZhengZhu24}. Moreover, the following results (which will be used in the examples in Section \ref{QECC_constructions}) present new quantum code constructions utilizing the lengthening, shortening, and dimension reduction of a given stabilizer quantum code.
\begin{lemma}\cite[Lemmas 69, 70, 71, and Corollary 73]{Ketkar-Klappenecker-Kumar-Sarvepalli2006}\label{qeecc1}
\begin{enumerate}
\item If an $[[n,k,d]]_q$ stabilizer code exists for $k> 0$, then there exists an impure $[[n+1, k,d]]_q$ stabilizer code.  
\item If a pure $[[n,k,d]]_q$ stabilizer code exists with $n\geq 2$ and $d\geq 2$, then there exists a $[[n-1, k+1, d-1]]_q$ pure stabilizer code. 
\item If a (pure) $[[n,k,d]]_q$ stabilizer code exists, with
$k\geq 2$ $(k\geq 1)$, then there exists an $[[n, k-1, d^*]]_q$ stabilizer code (pure to $d$) such that $d^* \geq d$.
\item Suppose that an $[[n,q^k ,d]]_q$ and an $[[n', q^{k'},d']]_q$ stabilizer code exist. Then there exists an \linebreak $[[n+n' , q^{k + k'}, \min\{d, d'\}]]_q$ stabilizer code.
\end{enumerate}
\end{lemma}

\section{A Construction of ESO/EDC/ESD QC Codes From Their Constituent Codes }\label{mainsection}

In this section, we discuss the main results of this paper where we propose a recursive construction of infinite families $\{\C_{u}\}_{u\geq 1}$ of good QC codes that are self-orthogonal (or dual-containing or self-dual) with respect to the Euclidean inner product. 
Such constructions of QC codes will be utilized to show the existence of QECCs in Section~\ref{QECC_constructions}.

\subsection{Recursive Construction of Infinite Families of QC Codes}

To provide a careful description of our construction, we first describe a general case of the construction and then specify certain conditions on the constituent codes to explicitly describe the code dimension and lower bound for its minimum distance. Let $x^m-1$ factor into $ 2t +s$ irreducible factors as in \eqref{decomposition_of_R^ell}, where $f_{s}=x-1$.
Then,
\begin{equation} \label{eq:new_r_decomp}
R^{\ell}_{q,m} \cong \bigoplus_{j=1}^t\left(\frac{\F_{q} [x]}{\langle g_j \rangle} \oplus \frac{\F_{q} [x]}{\langle g_j^* \rangle} \right)^{\ell} \oplus \left( \bigoplus_{i=1}^{s-1} \left(\frac{\F_{q} [x]}{\langle f_i \rangle} \right)^{\ell}  \right) \oplus \left(\frac{\F_{q}[x]}{\langle x-1 \rangle}\right)^{\ell}.
\end{equation}
%
%
Let $\mathcal{C}_1 \subseteq \mathbb{F}_{q}^{m\ell}$ be an $\ell$-QC code whose first $2t$ constituent codes are direct sums of arbitrary codes and their duals, i.e.,
\begin{equation*}\label{first_2t_level1}
\bigoplus_{j=1}^{t} \left(C'_{1,j}\oplus {C'}_{1,j}^{\perp_E}\right) \subset \bigoplus_{j=1}^{t}\left(\frac{\F_q [x]}{\langle g_j \rangle} \oplus \frac{\F_q [x]}{\langle g_j^* \rangle} \right)^\ell
\end{equation*}
and whose last $s$ constituent codes are 
\begin{equation*}\label{last_s_level1}
\left(\bigoplus_{i=1}^{s-1} C_{1,i}\right) \oplus C_{1,s} \subset \left(\bigoplus_{i=1}^{s-1}\left(\frac{\F_{q} [x]}{\langle f_i \rangle}\right)^\ell\right) \oplus \left(\frac{\F_{q} [x]}{\langle x-1\rangle} \right)^\ell
\end{equation*}
where the codes $C_{1,i}$ are HSO codes for $1 \leq i \leq s-1$ and $C_{1,s}$ is ESO (see Remark~\ref{euclidean_hermitian_remark}). Thus,
\begin{equation} \label{C1_decomp}
\mathcal{C}_1 \cong  \Bigg(\bigoplus_{j=1}^{t} \left(C'_{1,j}\oplus {C'}_{1,j}^{\perp_E}\right)\Bigg) \oplus \left(\bigoplus_{i=1}^{s-1} C_{1,i}\right) \oplus C_{1,s}.
\end{equation}
As $C_{1,s} \subset \F_{q}^{\ell}$, we can assume without much loss of generality that the it is the constituent code of $\C_{1}$ with the smallest minimum distance. Moreover, as the $C_{1,i}$ are HSO for each $1 \leq i \leq s-1$ and $C_{1,s}$ is ESO, Theorem~\ref{lingsole} tells us that $\mathcal{C}_{1}$ is an ESO $\ell$-QC code of length $m\ell$. Now using $\mathcal{C}_{1}$, we construct another QC code $\mathcal{C}_2$ as follows: Let
\begin{equation*}
R^{m\ell}_{q,m} \cong \bigoplus_{j=1}^t\left(\frac{\F_{q} [x]}{\langle g_j \rangle} \oplus \frac{\F_{q} [x]}{\langle g_j^* \rangle} \right)^{m\ell} \oplus \left( \bigoplus_{i=1}^{s-1} \left(\frac{\F_{q} [x]}{\langle f_i \rangle} \right)^{m\ell}  \right) \oplus \left(\frac{\F_{q}[x]}{\langle x-1 \rangle}\right)^{m\ell}.
\end{equation*}
Let $\mathcal{C}_2 \subseteq  \mathbb{F}_{q}^{m^2\ell}$ be the $m\ell$-QC code whose first $2t$ constituent codes are the following direct sums of arbitrary codes and their duals,
\begin{equation*}\label{first_t_level2}
\bigoplus_{j=1}^{t} \left(C'_{2,j}\oplus {C'}_{2,j}^{\perp_E}\right) \subset \bigoplus_{j=1}^{t}\left(\frac{\F_q [x]}{\langle g_j \rangle} \oplus \frac{\F_q [x]}{\langle g_j^* \rangle} \right)^{m\ell},
\end{equation*}
and whose last $s$ constituent codes are 
\begin{equation*}\label{last_s_level2}
\left(\bigoplus_{i=1}^{s-1} C_{2,i}\right) \oplus \mathcal{C}_1 \subset \left(\bigoplus_{i=1}^{s-1}\left(\frac{\F_{q} [x]}{\langle f_i \rangle}\right)^{m\ell}\right) \oplus \left(\frac{\F_{q} [x]}{\langle x-1\rangle} \right)^{m\ell}
\end{equation*}
where the $C_{2,i}$ are HSO codes for $1 \leq i \leq s-1$ and $\mathcal{C}_1 \subset \left( \frac{\mathbb{F}_{q}[x]}{\langle x-1 \rangle} \right)^{m\ell} \cong \mathbb{F}_{q}^{m\ell}$ is an ESO constituent code of $\mathcal{C}_2$. 
Thus,
%
\begin{equation*}
\mathcal{C}_{2}\cong \Bigg( \bigoplus_{j=1}^{t} \left(C'_{2,j}\oplus {C'}_{2,j}^{\perp_E}\right) \Bigg)\oplus \left( \bigoplus_{i=1}^{s-1} \; C_{2,i} \right) \oplus \mathcal{C}_1
\end{equation*}
Again by Theorem~\ref{lingsole}, $\mathcal{C}_{2}$ is an ESO $m\ell$-QC code of length $m^2\ell$. To extend this recursive construction, we have that any $u\geq 2$, 
\begin{equation*}
R^{m^{u-1}\ell}_{q,m} \cong \bigoplus_{j=1}^t\left(\frac{\F_{q} [x]}{\langle g_j \rangle} \oplus \frac{\F_{q} [x]}{\langle g_j^* \rangle} \right)^{m^{u-1}\ell} \oplus \left( \bigoplus_{i=1}^{s-1} \left(\frac{\F_{q} [x]}{\langle f_i \rangle} \right)^{m^{u-1}\ell}  \right) \oplus \left(\frac{\F_{q}[x]}{\langle x-1 \rangle}\right)^{m^{u-1}\ell}.
\end{equation*}
and define an $m^{u-1}\ell$-QC code $\mathcal{C}_{u}$ in $\mathbb{F}_{q}^{m^u \ell}$ 
whose first $2t$ arbitrary constituent codes are 
\begin{equation}\label{first_2t_general_level1}
\bigoplus_{j=1}^{t} \left(C'_{u,j}\oplus {C'}_{u,j}^{\perp_E}\right) \subset \bigoplus_{j=1}^{t}\left(\frac{\F_q [x]}{\langle g_j\rangle} \oplus \frac{\F_q [x]}{\langle g_j^*\rangle} \right)^{m^{u-1}\ell}
\end{equation}
and whose last $s$ constituent codes are 
\begin{equation}\label{last_s_general_level1}
\left( \bigoplus_{i=1}^{s-1} \;C_{u,i} \right) \oplus \mathcal{C}_{u-1} \subset \left( \bigoplus_{i=1}^{s-1} \left(\frac{\F_{q} [x]}{\langle f_i\rangle} \right)^{m^{u-1}\ell}  \right) \oplus \left(\frac{\F_{q}[x]}{\langle x-1 \rangle}\right)^{m^{u-1}\ell},
\end{equation}
where the $C_{u,i}$ are HSO codes for all $1 \leq i \leq s-1$ and $\mathcal{C}_{u-1} \subset \left( \frac{\mathbb{F}_{q}[x]}{\langle x-1 \rangle} \right)^{m^{u-1}\ell} \cong \mathbb{F}_{q}^{m^{u-1}\ell}$ is an ESO constituent code of $\mathcal{C}_u$. So by Theorem~\ref{lingsole}, 
\begin{equation}\label{decomposition_C_u}
\mathcal{C}_u \cong\Bigg( \bigoplus_{j=1}^t\left( C_{u,j}'\oplus {C'}_{u,j}^{\perp_E} \right) \Bigg) \oplus \left(\bigoplus_{i=1}^{s-1} C_{u,i} \right) \oplus \mathcal{C}_{u-1},
\end{equation}
is an ESO $m^{u-1}\ell$-QC code of length $m^u\ell$ for all $u \geq 2$.

\begin{remark}
    Since the dual of an ESO code is EDC, our codes $\C_u ^{\perp}$ will be an EDC $m^{u-1}\ell$-QC code of length $m^u\ell$ for all $u\geq 2$, in which case Proposition~\ref{lingsole-decomposition2} is utilized.
\end{remark}

\subsection{The Dimension of the Family of Codes $\{\C_u\}_{u \geq 1}$}\label{section:dimension}

In the next lemma, we determine a formula for the dimension of each $\mathcal{C}_u$ in our recursive construction which depend exclusively on the dimensions of constituent codes of $\C_1$. To compute the dimension, we restrict to the special case in which the HSO constituent codes of $\mathcal{C}_u$ are obtained by taking $m^{u-1}$-fold copies of the HSO constituent codes of $\mathcal{C}_1$; that is, $C_{u,i} = m^{u-1}C_{1,i}$ for all $1 \le i \le s-1$. This choice allows us to systematically control how the parameters of $\mathcal{C}_u$ evolve with $u$, since taking copies preserves the dimension while scaling the length and minimum distance in a predictable way; namely, since each $C_{1,i}$ is HSO, it follows from a generalization of Part~(i) of Theorem~\ref{constbasedonprevioussocode} that each $C_{u,i}$ is also HSO. Moreover, if $C_{1,i}$ has parameters $[\ell, k_{1,i}, d_{1,i}]$, then $C_{u,i}$ has parameters $[m^{u-1}\ell,\, k_{1,i},\, m^{u-1}d_{1,i}]$ for all $1 \le i \le s-1$. While this assumption restricts us to a special case, it provides a tractable setting for analysis, and extensions to more general constructions will be explored in future work.
\begin{lemma}\label{dimension_lemma_Cu1}
Let $R_{q,m}^{\ell}$ be as in \eqref{eq:new_r_decomp}. Then for $u\geq 1$, there exists an $m^u \ell$-length QC code $\C_u$ with constituent codes~\eqref{first_2t_general_level1} and \eqref{last_s_general_level1}, whose dimension is
\begin{equation*}
\dim(\mathcal{C}_u) = \ell\left(\frac{m^{u} -1}{m-1}\right)\left(\sum_{j=1} ^t \deg(g_j)\right) + u\left(\sum_{i=1} ^{s-1} \deg(f_i) k_{1,i}\right)+k_{1,s}.
\end{equation*}
where $k_{1,i}=\dim(C_{1,i})$, for all $1 \leq i \leq s$.
\end{lemma}
\begin{proof} 
Let $\C_1$ be the QC code with decomposition on constituent codes as
 
\begin{equation*}\label{C1}
\mathcal{C}_1 \cong \left(\bigoplus_{j=1}^{t} \left({C_{1,j}^{'}} \oplus {C_{1,j} ^{'}}^{\perp_E}\right) \right) \oplus \left(\bigoplus_{i=1}^s C_{1,i} \right).
\end{equation*}
Using Lemma~\ref{dimensionlemma} and the facts that $\deg(g_{j})=\deg(g_j^*)$ and $\dim({C'}_{1,j}^{\perp_E})=\ell-k'_{1,j}$ for all $1 \leq j \leq t$, we have that
\begin{eqnarray*}\label{dimension_C_1}
\dim(\C_1)&=& \left(\sum_{j=1}^t \deg(g_j)k_{1,j} ' + \deg(g_j ^*)  (\ell - k_{1,j} ')\right) + \left( \sum_{i=1}^{s} \deg(f_i)k_{1,i}\right) \nonumber \\
&=&\ell\sum_{j=1}^t \deg(g_j) + \left( \sum_{i=1}^{s-1} \deg(f_i)k_{1,i}\right)+k_{1,s}.
\end{eqnarray*}
For any $u\geq 2$, let $\C_u$ be the QC code defined in \eqref{decomposition_C_u}
where the $C_{u,i}$ are $m^{u-1}$-copies of $C_{1,i}$, namely, $C_{u,i} = m^{u-1}C_{u,i}$, for all $1 \leq i \leq s$. Given its recursive argument, $\dim(\mathcal{C}_u)$ is computed via induction. The case when $u=1$ is done. Now assume that the argument is valid for any integer $u$, i.e., there exists a QC code $\mathcal{C}_u \subset \F_{q} ^{m^u \ell}$ with the following dimension
\begin{equation}\label{dimension_C_u}
\dim(\mathcal{C}_u) = \ell\left(\frac{m^u-1}{m-1} \right)\left(\sum_{j=1}^{t}\deg(g_j)\right)+u \left(\sum_{i=1}^{s-1} \deg(f_{i}) k_{1,i}\right) + k_{1,s} 
\end{equation}
where $k_{1,i}= \dim(C_{1,i})=\dim(m^{u-1}C_{1,i})=\dim(C_{u,i})$ for all $1 \leq i \leq s-1$. Now let $\C_{u+1} \subset \F_{q} ^{m^{u+1}\ell}$ be the QC code with the following $m^{u} \ell$-length constituent codes
\begin{equation*}\label{decompositionproof}
 \C_{u+1} \cong  \left(\bigoplus_{j=1}^{t} \left( C_{u+1, j} ' \oplus C_{u+1,j} ^{' \perp_E}\right) \right) \oplus \left(\bigoplus_{i=1}^{s-1} {C_{u+1,i}}\right) \oplus \C_{u}
\end{equation*}
where each $C_{u+1,i}=m^{u} C_{1,i}$ are also HSO for all $1 \leq i \leq s-1$ (see Theorem~\ref{constbasedonprevioussocode}). Using~\eqref{dimension_C_u}, we have
%
%
\begin{align*}
\dim(\mathcal{C}_{u+1}) &= \sum_{j=1}^t k'_{u+1,j}\deg(g_j)+\sum_{j=1}^t \deg(g_{j}^*){k'}_{u+1,j}^{\perp_E}+\left(\sum_{i=1}^{s-1}\deg(f_i) k_{1,i}\right)+\dim(\mathcal{C}_{u}) \\
& = \sum_{j=1}^t \left[ \deg(g_{j}) k'_{u+1,j} + \deg(g_{j})\left(m^u\ell-k'_{u+1,j}\right) \right] + \left(\sum_{i=1}^{s-1}\deg(f_i) k_{1,i}\right)+\dim(\mathcal{C}_{u})\\
& = m^u\ell \sum_{j=1}^t\deg(g_j )+\left(\sum_{i=1}^{s-1}\deg(f_i )k_{1,i}\right)+\ell\left(\frac{m^u-1}{m-1} \right)\left(\sum_{j=1}^{t}\deg(g_j )\right)+u\left(\sum_{i=1}^{s-1} \deg(f_{i} ) k_{1,i}\right) + k_{1,s}\\
&=\ell\left(\frac{m^{u+1} -1}{m-1}\right)\left(\sum_{j=1} ^t \deg(g_j)\right) + (u+1)\left(\sum_{i=1} ^{s-1}\deg(f_i) k_{1,i} \right)+k_{1,s}
\end{align*}
as desired. \qed
\end{proof}

\subsection{A Lower Bound for the Minimum Distance of the Family of Codes $\{\C_u\}_{u \geq 1}$}

In this section, we will use \textit{Jensen's bound} to find a lower bound for the minimum distance of the QC code $\C_u \subset \F_q ^{m^u \ell} $ for each $u \geq 1$. 
\color{black}
In order to use Theorem \ref{jensenbound} in conjunction with our recursive construction of $\C_u$ outlined in \eqref{first_2t_general_level1} and \eqref{last_s_general_level1}, the constituent codes of each $\C_u$ must satisfy the descending chain condition in~\eqref{chaincond}. We first decompose $\C_u$ for $u\geq 2$ as
\[ \C_u \cong \bigoplus_{\lambda=1}^{2t+s} C_{u,\lambda}\]
where the $C_{u,\lambda}$ are the nontrivial constituent codes of $\C_u$ whose minimum distances satisfy~\eqref{chaincond}, noting that $\C_{u-1} = C_{u,\lambda}$ for some $\lambda$. In general, we do not have enough information to conclude how $d(\C_{u-1})$ compares with the minimum distances of the other constituent codes for $\C_u$ and do not know where it appears in the descending chain. Thus, we require the additional condition that the constituent code with the smallest minimum distance in $\C_u$ grows by a factor of $m$ from the smallest constituent code in $\C_{u-1}$, i.e.,
\begin{equation}\label{constituent_code_condition2}
d(\C_{u,2t+s})\ge md(\C_{u-1,2t+s}).
\end{equation}
This is not a major assumption as $\C_u \subset \F_q ^{m^u \ell}$ and $\C_{u-1} \subset \F_q ^{m^{u-1} \ell}$, so its minimum distance can be a factor of $m$ larger than $\C_{u-1}$. Repeated application of~\eqref{constituent_code_condition2} yields
\begin{equation}\label{constituent_code_condition3}
d(\C_{u,2t+s})\ge m^{u-1}d_{1,s}
\end{equation}
where $d_{1,s}=d(C_{1,s})$ in~\eqref{C1_decomp}.
We emphasize that having such a condition will be vital in proving a lower bound for the minimum distances of $\C_u$ for each $u$ in the following lemma. This condition can be satisfied, for instance, by setting $C_{u,\lambda}=m C_{u-1,\lambda}$ for each $u\geq 2$ and $1 \leq \lambda \leq 2t+s$. However, unlike in Lemma~\ref{dimension_lemma_Cu1}, taking $m$ copies of constituent codes is not required to obtain our result.

\begin{lemma}\label{minimumdistance_lemma_Cu1}
Let $R_{q,m}^{\ell}$ be as in \eqref{eq:new_r_decomp}.  Then for $u\geq 1$, there exists an $m^u \ell$-length QC code $\C_u$ with constituent codes~\eqref{first_2t_general_level1} and \eqref{last_s_general_level1}, with minimum distance $d(\mathcal{C}_u) \geq m^{u-1}d_{1,s}$, where $d_{1,s}=d(C_{1,s})$.
\end{lemma}

\begin{proof} The case for $\C_1$ is trivial as we are assuming that $C_{1,s}$ has the smallest minimum distance in the representation of $\C_1$ in~\eqref{C1_decomp}. For each $u\geq2$, we will assume that the constituent codes in the decomposition of $\C_u$ are ordered according to the descending chain condition~\eqref{chaincond}. The proof will require that we traverse our table horizontally by way of the descending chain condition to the last constituent then up the table vertically to use the inequality \eqref{constituent_code_condition3}. For each $u\geq 2$, using the descending chain condition~\eqref{chaincond} and inequality \eqref{constituent_code_condition3}, we have the following:
\begin{align*}
N_{2t+s} (\C_u) &= d(C_{u,2t+s})d(D_{2t+s}) \ge m d(C_{u-1,2t+s}) \ge \ldots \ge m^{u-1} d(C_{1,s}) = m^{u-1} d_{1,s},\\
N_{2t+s-1,2t+s} (\C_u) &= d(C_{u,2t+s-1})d(D_{2t+s-1,2t+s})\geq d(C_{u,2t+s})\geq m^{u-1} d(C_{1,s}) = m^{u-1} d_{1,s},\\
& \vdots \\
N_{1,2,...,2t+s}(\C_u) &= d(C_{u,1})d(D_{1,2,...,2t+s}) \ge d(C_ {u,2t+s}) \geq m^{u-1} d(C_{1,s}) = m^{u-1} d_{1,s}.
\end{align*}
Thus, $d(\C_u) \geq d_{\mathrm{J}}(\C_u) = \min\{N_{2t+s}(\C_u), N_{2t+s-1,2t+s}(\C_u),\ldots, N_{1,2,\ldots,2t+s}(\C_u)\} \geq m^{u-1}d_{1,s}$ as desired. \qed
\end{proof}

\begin{example} For sake of clarity, consider the following table of the decompositions of $\C_u$ for $2 \leq u \leq 4$, where the decomposition is found in~\eqref{decomposition_C_u} and the constituent codes satisfy~\eqref{constituent_code_condition2}.
\begin{center}
\begin{tabular}{|c|c|}
\hline 
$u$ & Decomposition \\
\hline 
$2$ & $\mathcal{C}_2 \cong C_{2,1} \oplus C_{2,2} \oplus \cdots \oplus  C_{2,2t+s}$    \\
\hline
$3$ & $\mathcal{C}_3 \cong C_{3,1} \oplus C_{3,2} \oplus \cdots \oplus  C_{3,2t+s}$   \\
\hline
$4$ & $\mathcal{C}_4 \cong C_{4,1} \oplus C_{4,2} \oplus \cdots \oplus  C_{4,2t+s}$\\
\hline 
\end{tabular}
\captionof{table}{Decompositions of $C_2, \C_3, \C_4$ to obtain lower bound for $d(\C_4)$}
\end{center}
Traversing horizontally through the row, we have that for $I = \{r,r+1,\ldots,2t+s\} \subseteq \{1,2,\ldots, 2t+s\}$, $N_I(\C_4) \ge d(C_{4,r}) \ge d(C_{4,2t+s})$ by the descending chain condition, and then by equation ~\eqref{constituent_code_condition2}, traversing vertically up the table along the last constituent codes yields $N_I(\C_4) \ge m^3 d(C_{1,s}) = m^3 d_{1,s}$ for all $I$. Hence, $d(\C_4) \geq m^3d_{1,s}.$
\end{example}

\color{black}
The following theorem provides a summary of our construction of a family of ESO QC codes in $\mathbb{F}_{q}^{m^{u}\ell}$.
\begin{theorem}\label{maintheorem1}
 Let $R_{q,m}^{\ell}$ be as in \eqref{eq:new_r_decomp}. Let $\mathcal{C}_{u}$ be the $m^u\ell$-length QC code whose constituent codes are in \eqref{first_2t_general_level1} and \eqref{last_s_general_level1}, where $C_{u,i}=m^{u-1}C_{1,i}$, for all $1 \leq i \leq s-1$, and that their minimum distances satisfy~\eqref{chaincond}. Then $\{\mathcal{C}_u\}_{u\geq 1}$ is an infinite family of ESO $q$-ary QC codes with parameters 
\begin{equation}
\left[m^u \ell ,\ell\left(\frac{m^{u} -1}{m-1}\right)\left(\sum_{j=1} ^t \deg(g_j)\right) + u\left(\sum_{i=1} ^{s-1} \deg(f_i) k_{1,i}\right)+k_{,1s}, \geq m^{u-1} d_{1,s}\right]_q.
\end{equation}
\end{theorem}

\begin{proof}
The parameters and duality of $\C_u$ comes from the recursive construction, and Lemmas~\ref{dimension_lemma_Cu1} and~\ref{minimumdistance_lemma_Cu1}. \qed
\end{proof}

The next result is a particular case of Theorem~\ref{maintheorem1} which allow us also to construct infinite families of not only ESO (or EDC) QC codes, but also QC codes that are ESD. From now on, we assume $m$ is odd when $q$ is odd (see Remark~\ref{euclidean_hermitian_remark}) so that the unique self-reciprocal irreducible factor of $x^m-1$ is $x-1 \in \F_q [x]$. Hence in these special cases, we no longer need to use copies of the constituent codes of $\C_u$ to obtain the constituent codes of $\C_{u+1}$. Consequently, as we no longer have the dimensional restriction of $\dim(C_{u,i})=\dim(C_{1,i})=k_{1,i}$ for each $1\leq i \leq s-1$, we may allow for the $\C_u$ to ESD. Note that $\ell$ must be even in order for $\C_u$ to be ESD.

\begin{corollary}\label{cormainTh} 
Let $\displaystyle R_{q,m}^{\ell}$ be as in \eqref{eq:new_r_decomp} where the only self-reciprocal factor of $x^m-1$ is $x-1$. Let $\C_1 \subset \F_q ^{m\ell}$ be the ESO (resp. EDC or ESD) $\ell$-QC code whose constituent codes are $C_{1,j} \oplus C_{1,j} ^{\perp_E} \subset \left(\frac{\F_q [x]}{\langle g_j \rangle} \oplus \frac{\F_q [x]}{\langle g_j ^* \rangle}\right)^{\ell}$ for $1 \leq j \leq t$, and an ESO (resp. EDC or ESD)  code $C_{1,t+1} \subset \left(\frac{\F_q [x]}{\langle x-1\rangle }\right)^{\ell}$ with parameters $[\ell, k_{1,t+1}, d_{1,t+1}]_q$. Suppose that $d(C_{1,1})\geq d(C_{1,1}^{\perp_E})\geq \cdots\geq d(C_{1,t})\geq d(C_{1,t}^{\perp_E})\geq d(C_{1,t+1}).$ Then there exists an infinite family $\{\C_u\}_{u\geq 1}$ of ESO (resp. EDC or ESD) $q$-ary QC codes with parameters 
\begin{equation}
\left[m^u \ell, \ell\left(\frac{m^{u} -1}{m-1}\right)\left(\sum_{j=1} ^t \deg(g_j)\right) +k_{1, t+1}, \geq m^{u-1} d_{1,t+1}\right]_q.
\end{equation}
\end{corollary}

\begin{proof} 
This is a special case of Theorem~\ref{maintheorem1} where $s=1$, and the unique self-reciprocal irreducible factor of $x^m -1$ is $x-1$. Hence, 
\begin{equation}
\sum_{j=1} ^{t} \deg(g_j)=\frac{1}2\left[(m-\left(\sum_{i=1}^{s-1}\deg(f_i) \right)-1\right] = \frac{1}{2}(m-1), 
\end{equation}
namely, $\displaystyle{\dim \C_u = \ell\left(\frac{m^{u} -1}{m-1}\right)\left(\sum_{j=1} ^{t} \deg(g_j)\right) +k_{1, t+1} = \frac{\ell}{2}\left(m^{u} -1\right) +k_{1, t+1}}$. 

The duality of all constituent codes may be taken under the Euclidean inner product (see Remark~\ref{euclidean_hermitian_remark} - item (i)) and it is defined exclusively by the $q$-ary constituent code $C_{1,t+1}$. Indeed, if $C_{1,t+1}$ is ESO (resp. EDC or ESD), then $\C_1$ also is. By the recursive construction, for any $u\geq 1$, the duality of $\C_1$ defines the duality of $\C_2$, and so on. This avoids relying on copies of constituent codes, which would otherwise restrict the construction to only ESO (or EDC) codes because of the dimensional limitations of self-dual and dual-containing codes. \qed
\end{proof}

We note that if we want $x^m-1$ to have three distinct irreducible factors in $\F_q[x]$, of which $x-1$ is one, we need that $m$ is prime or the square of a prime. Otherwise, $m$ would have at least two distinct prime factors so that in the factorization
\begin{equation}\label{factorization}
x^m - 1 = \prod_{d\mid m} \Phi_d(x),
\end{equation}

\noindent where $\Phi_d(x)$ denotes the $d$-th cyclotomic polynomial, there would be at least four cyclotomic polynomial factors ($1$, $m$, and the distinct prime factors of $m$). This would result in at least four irreducible factors.

If $m$ is prime,
\begin{equation}
x^m-1 = \Phi_1(x) \Phi_m(x) = (x-1)\Phi_m(x),
\end{equation}

\noindent so we must consider the irreducible factors of $\Phi_m$. The polynomial $\Phi_m$ has $\frac{\phi(m)}{\operatorname{ord}_m(q)} = \frac{m-1}{\operatorname{ord}_m(q)}$ irreducible factors. 
Thus, if $x^m-1$ has three irreducible factors, we need $\operatorname{ord}_m(q) = \frac{m-1}{2}$, i.e. $q^{(m-1)/2} \equiv 1 \Mod{m}$.

If $m$ is the square of a prime, say $p^2$ where $p$ is prime,
\begin{equation}
x^m-1 = \Phi_1(x) \Phi_p(x) \Phi_{p^2}(x).
\end{equation}

\noindent We need that both $\Phi_p$ and $\Phi_{p^2}$ are both irreducible, i.e., $1 = \frac{\phi(p)}{\operatorname{ord}_p(q)} = \frac{p-1}{\operatorname{ord}_p(q)}$ so that $\operatorname{ord}_p(q) = p-1$ and $1 = \frac{\phi(p^2)}{\operatorname{ord}_{p^2}(q)} = \frac{p(p-1)}{\operatorname{ord}_{p^2}(q)}$ so that ${\operatorname{ord}_{p^2}(q)} = p(p-1)$. Tables~\ref{table1_pairs_of_reciprocals_x-1} and~\ref{table2_pairs_of_reciprocals_x-1} provide some examples of $m$ and $q$ pairs where $x^m -1 \in \F_{q} [x]$ decomposes into one pair of reciprocal polynomials and $x-1$. One caveat is that over $\F_{q^2}[x]$, i.e., when $q$ is the square of a prime, $\Phi_p$ is irreducible if and only if $p=2$ by Fermat's little theorem, and so $x^{p^2}-1$ does not have three irreducible factors over $\F_{q^2}$.

\begin{center}
\begin{tabular}{|c|c|}
\hline
$q$ & $1\leq m \leq 100$               \\ \hline
2     & 7, 17, 23, 41, 47, 71, 79, 97  \\ \hline
3     & 11, 23, 37, 47, 59, 71, 83, 97 \\ \hline
5     & 11, 19, 29, 41, 59, 61, 79, 89 \\ \hline
7     & 3, 31, 47, 53, 59, 83          \\ \hline
11    & 7, 53, 79, 83, 97              \\ \hline
\end{tabular}
\captionof{table}{\label{table1_pairs_of_reciprocals_x-1} Some values of $m$ and $q$ for which $x^m -1\in \F_q[x]$ decomposes into three irreducible factors over $\F_q$.}
\bigskip
\begin{tabular}{|c|c|}
\hline
$q^2$ & $1\leq m \leq 100$             \\ \hline
4     & 3,5,7,11,13,19,23,29,37,47,53,59,61,67,71,79,83              \\ \hline
9     & 5,7,11,17,19,23,29,31,43,47,53,59,71,79,83,89\\ \hline
25    & 3,7,11,17,19,23,37,43,47,53,59,73,79,83,97\\ \hline
49    & 3,5,11,13,17,23,31,41,47,59,61,67,71,79,83,89,97             \\ \hline
121   & 3,7,13,17,23,29,31,41,47,59,67,71,73,79,83              \\ \hline
\end{tabular}
\captionof{table}{\label{table2_pairs_of_reciprocals_x-1} Some values of $m$ and $q^2$ for which $x^m -1 \in \F_{q^2}[x]$ decomposes into three irreducible factors over $\F_{q^2}$.}
\end{center}

\begin{example}
Let 
\begin{eqnarray*}
R_{3,11} ^{5}&\cong&\left(\frac{\F_3 [x]}{\langle x^5 + 2x^3 +x^2 +2x+2\rangle}\right)^5\oplus \left(\frac{\F_3 [x]}{\langle x^5 + x^4 +2x^3 +x^2 +2 \rangle }\right)^5 \oplus \left(\frac{\F_3 [x]}{\langle x-1 \rangle }\right)^5 \\
&=&\left(G_1 '\oplus G_1 ''\right)^5 \oplus \left( \F_3\right)^5 ,
\end{eqnarray*}
where $G_1 ' = \F_3 (\alpha)$. In addition, let $C_{1,1} ' = \langle (1,2,1,2,1), (\alpha, \alpha^2, \alpha^3, \alpha^4,\alpha^5 ) \rangle$ be the $[5, 2, 4]_{3^5}$-code, $C_{1,1} '' = C_{1,1} ^{' \perp_E} \subset (G_1 '')^{5}$ be the $[5,3,3]_{3^5}$-code, and $C_{11} = \langle (1,1,1,0,0) , (1,2,0,1,0)\rangle $ be the self-orthogonal $[5,2,3]_3 $-code. Let these codes be are the constituent codes of the QC code $\C_1$. For any $u\geq 2$, let the first two constituent codes of $\C_u$ be multiples of $\C_{1,1}'$ and $C_{1,1} ^{' \perp_E}$ to obtain the $m^u \ell$-length QC codes
\begin{eqnarray}
\C_u \cong 11^{u-1}C_{1,1} ' \oplus 11^{u-1}C_{1,1} ^{' \perp_E} \oplus \C_{u-1}.
\end{eqnarray}
According to Corollary~\ref{cormainTh}, the $\{\mathcal{C}_u\}_{u\geq 1}$ is an infinite family of ESO QC codes with parameters $$\left[5\cdot 11^u , \frac{5\cdot 11^u -1}{2} , \geq 3\cdot 11^{u-1} \right]_3.$$
\end{example}

\section{Satisfying the Square-Root-Like Lower Bound}\label{section:sqrtbound}

Using results from~\cite{lingsole} and developments presented in this paper, we show how the construction of codes in Section~\ref{mainsection}  satisfy the square-root-like lower bound and then present explicit examples. We first state the following helpful lemma.
\begin{lemma}\label{lemmaforinffamsatisfyingsquarerootbound}
If $m \ge 1$ and $u \ge 2$ are positive integers, then $\sqrt{m^u} \le  m^{u-1}.$
\end{lemma}

\begin{proof}
If $u\ge 2$ , then $u \le 2u-2$. For $m\ge 1$, this gives $m^u \le m^{2u-2}$ and so $\sqrt{m^u} \le m^{u-1}$. \qed
\end{proof}

\begin{theorem}
Let $\{\C_u\}_{u\geq 1}$ be the infinite family of ESO $q$-ary QC codes proposed in Theorem~\ref{maintheorem1}. Then
\begin{itemize}
    \item[(i)] If $d(C_{1,s}) \ge c\sqrt{\ell}$ for some constant $c$, then all codes in $\{\C_u\}_{u\geq 2}$ satisfy the square-root-like bound for their minimum distances.
    \item[(ii)] Otherwise, there exist a positive integer $L >1 $ so that the all codes in the subfamily $\{\C_u\}_{u\geq L}$ of $\{\C_u\}_{u\geq 1}$ satisfies the square-root-like bound for their minimum distances.    
\end{itemize}
\end{theorem}
\begin{proof}
\begin{itemize}
    \item[(i)] By Lemma~\ref{lemmaforinffamsatisfyingsquarerootbound}, for each $m^u \ell$-length QC code $\C_u$ with $u\geq 2$, we have 
\begin{equation*}
d(\C_u)\geq m^{u-1}d_{1,s}\geq \sqrt{m^u}d_{1,s}\geq \sqrt{m^u}c\sqrt{\ell}=c\sqrt{m^u \ell},   
\end{equation*}
namely, the minimum distances in $\{\C_u \}_{u\geq 2}$ satisfy the square-root-like bound and the result follows.
\item[(ii)] Notice that 
\[m^{u-1} d_{1s} \geq c\sqrt{m^u \ell} \iff  m^{2u -2} d_{1s} ^2 \geq c^2 m^u \ell \iff m^{u-2} \geq \frac{c^2 \ell}{d_{1s} ^2}.\] 
Hence,
$u-2 \geq \log_{m} \frac{c^2 \ell}{d_{1s}^2}$, and when $u \geq \left\lceil\log_{m} \frac{c^2 \ell}{d_{1s}^2} \right\rceil +2=: L$, the subfamily $\{\C_u\}_{u\geq L}$ of $\{\C_u\}_{u\geq 1}$ satisfies the square-root-like bound.  \qed  
\end{itemize}
\end{proof}

In the special case where 
$\displaystyle R_{q,m}^{\ell}$ in \eqref{eq:new_r_decomp} has only $x-1$ as the self-reciprocal factor of $x^m-1$, we obtain EDC and ESD codes that satisfying the square-root-like bound.

\begin{corollary}\label{inffamilysquareroot}
Let $\{\C_u\}_{u\geq 1}$ be the infinite family of ESO/EDC/ESD $q$-ary QC codes proposed in Corollary~\ref{cormainTh}.
Then
\begin{itemize}
    \item[(i)] If $d(C_{1,t+1}) \ge c\sqrt{\ell}$ for some constant $c$, then all codes in $\{\C_u\}_{u\geq 2}$ satisfy the square-root-like bound for their minimum distances.
    \item[(ii)] Otherwise, there exist a positive integer $L >1 $ so that the all codes in the subfamily $\{\C_u\}_{u\geq L}$ of $\{\C_u\}_{u\geq 1}$ satisfies the square-root-like bound for their minimum distances.
\end{itemize}
\end{corollary}

\begin{remark}
Any $c \le m^{u/2-1}$ will work for square-root-like bound. If we would like a uniform $c$, then $c < \sqrt{m}$ works for all $u\geq 1$. 
\end{remark}

\begin{remark}
    Additionally, is straightforward to see that the lower-bound for the minimum distance in Theorem~\ref{maintheorem1} is asymptotically better than the square-root-like lower bound. Indeed, given that $d_{1s}, m,c$, and $\ell$ are positive integers, we have
\begin{equation*}
\lim_{u \rightarrow +\infty} \frac{m^{u-1}d_{1s}}{c\sqrt{m^u \ell}} =   \lim_{u \rightarrow +\infty} \frac{m^{\frac{u}{2}-1}d_{1s}}{c \sqrt{\ell}} = +\infty > 1.
\end{equation*}

\end{remark}

\begin{example}\label{exinffamilysqroot}
Let 
\begin{eqnarray*}
R_{5,11} ^{6}&\cong&\left(\frac{\F_5 [x]}{\langle x^5 + 2x^4 +4x^3 +x^2 +x+4\rangle}\right)^6 \oplus \left(\frac{\F_5 [x]}{\langle x^5 +4x^4 +4x^3 +x^2 +3x+4 \rangle }\right)^6 \oplus \left(\frac{\F_5 [x]}{\langle x-1 \rangle }\right)^6 \\
&=&\left(G_1 '\oplus G_1 ''\right)^6 \oplus \left( \F_5\right)^6 ,
\end{eqnarray*}
where $G_1 ' = \F_5 (\alpha)$. Let $C_{1,1} '$ be a $[6, 3, 4]_{5^5}$-GRS code and $C_{1,1}^{''} = C_{1,1} ^{' \perp_E}$, then $C_{11} ''$ is also a $[6,3,4]_{5^5}$-GRS code. In addition, let $C_{1,1} = \langle (1,0,0,2,2,4),(0,1,0,2,4,2),(0,0,1,4,2,2)\rangle$ be a self-dual $[6,3,4]_5 $ code. Observe that all constituent codes are MDS. Therefore, according to Corollary~\ref{inffamilysquareroot}, it is possible to get a infinite family $\{\mathcal{C}_u\}_{u\geq 2}$ of ESD QC codes with parameters $\left[4\cdot 11^u , 2\cdot 11^u , \geq 4\cdot 11^{u-1} \right]_5$ for $u\geq 2$ that satisfy the square-root-like bound for their minimum distances. 
\end{example}

\begin{example}\label{example_optimal2}
Let $\alpha$ be a primitive $7$-th root of unity and let 
\begin{eqnarray*}\label{CRT_decomposition_R2^3}
R_{2,7}^8 &\cong& \left(\frac{\mathbb{F}_2 [x]}{\langle x^3 +x +1 \rangle} \right)^8 \oplus\left(\frac{\mathbb{F}_2 [x]}{\langle x^3 +x^2 +1 \rangle}\right)^8\oplus\left(\frac{\mathbb{F}_2 [x]}{\langle x -1 \rangle}\right)^8 \nonumber \\
&\cong & \mathbb{F}_2 ^8(\alpha) \oplus \mathbb{F}_2 ^8 (\alpha^3 ) \oplus \mathbb{F}_2 ^8.
\end{eqnarray*}
Let $C_{1,1}'$ be an $[8,3,6]_8$ MDS code (whose construction can be checked in~\cite{table}), $C_{1,1} ^{''} = C_{1,1}^{'\perp_E}$ be an  $[8,5,4]_8$ (also MDS) code, and $C_{1,1}$ be an EDC $[8,6,2]_2 $-code with basis 
\begin{align*}
 B = \{ &(1,0,0,1,0,0,0,0), (0, 1, 0, 1, 0, 0, 0, 0), (0, 0, 1, 1, 0, 0, 0, 0), (0, 0, 0, 0, 1, 0, 0, 1), (0,0,0,0,0,1,0,1), \\
 & (0, 0, 0, 0, 0, 0, 1, 1)\}.
\end{align*}
%
%
%
Note that $d(D_3)=7$, $d(D_{2,3})=3$, and $d(D_{1,2,3})=1$. So from Lemma~\ref{dimensionlemma} and Theorem~\ref{jensenbound}, the code $\C_1$ is a $[56,30,\ge 6]_2$ $8$-QC code.
By Corollary~\ref{cormainTh}, we obtain a  $\left[ 8\cdot 7^u, 4\cdot 7^u+2, \geq 2\cdot 7^{u-1}\right]_2$ QC code.

When $u\geq 3$, $d \geq \sqrt{8\cdot7^u}$, which shows that the infinite family $\{C_u\}_{u \geq 3}$ satisfies the square-root bound for the minimum distances. Moreover for $c=5$, we have $\left\lceil\log_{m} \frac{c^2 \ell}{d_{1s}^2} \right\rceil +2= 5$. Thus by Corollary~\ref{inffamilysquareroot}, we then obtain an infinite family of QC codes $\{\C_u\}_{u\geq 5}$ whose minimum distances are greater than $c\sqrt{8\cdot 7^u}$. 
\end{example}

We conclude this section by emphasizing that, alongside the seminal constructions of infinite families of self-dual cyclic codes satisfying the square-root-like bound in~\cite{infiniteselfdualfamilies}, our work introduces a novel construction of infinite families of QC codes that also attain this square-root-like lower bound on minimum distance.

\color{black}

\section{Galois Closed QC Codes and Their Duals}\label{Galois-Closed-Section}

Having QC codes whose Hermitian and Euclidean duals coincide is a valuable property, as it allows us to move freely between the two inner products. To achieve this property, we introduce the notion of the Galois closure of a code $C$ using notation from~\cite{Jin2010}: Given a positive integer $m\geq 1$ and $c=(c_1 , c_2 , ..., c_n)\in C \subset \mathbb{F}_q ^n$ (in this case, $C$ is just a subset), we define $c^m = (c_1 ^m, c_2 ^m, ..., c_n ^m)$ and $C^m = \{c^m : c \in C \} \subset \mathbb{F}_q ^n.$ If $C$ is a subspace of $\mathbb{F}_{q}^n$ and $m$ is a power of the characteristic of $\mathbb{F}_{q}$, i.e., $m = q ^{t}$ for some $t \in \mathbb{Z}^+$, then $C^m$ is also a subspace of $\mathbb{F}_{q}^n$.

Let $\F_q \subset \F_{q^r}$ and $C$ be an $\F_{q^r}$-code. We say that $C$ is \textit{Galois closed} if $C = C^q$, where further on such codes details can be found in~\cite{Bierbrauer}.
The Frobenius map $C\mapsto C^q$ respects the quasi-cyclic property of $C$ so that $C^q$ is also quasi-cyclic whenever $C$ is. Furthermore, if $\phi(C) , \phi(C^q) \subset \left(\frac{\F_{q^2} [x]}{\langle x^m -1\rangle }\right)^{\ell}$, then we can look at the constituent codes of $C$ and $C^q$ and assume that the irreducible factors of $x^m -1$ are aligned for both decompositions.

\begin{remark}\label{codeC^q}
%
For any code $C$ over $\mathbb{F}_{q^2}$, we observe that the Hermitian dual $C^{\perp_H}$ is equal to the Euclidean dual $\left(C^q\right){^{\perp_E}}$ of $C^q$. Hence, $C$ is HSO if and only if $C \subset (C^q)^{\perp_E}$, i.e., $C^{q} \subset C^{\perp_{E}}$. Moreover, if $C$ is Galois closed, then $C^{\perp_{H}}=C^{\perp_{E}}$. Thus, $C$ is an HSO (resp. HDC, HSD) code if and only if $C$ is ESO (resp. EDC, ESD).   
\end{remark}

The following two propositions give a classification and construction for the Galois closed codes in $\mathbb{F}_{q^2}^{\ell}$, namely that they are direct sums of one-dimensional subspaces spanned by vectors with entries $0$ or $v_i$ satisfying $v_i^{q-1} = \beta$, for some fixed $\beta\in \mathbb{F}_{q^2}$. This allows for easy construction of Galois closed codes and demonstrates that they are not rare but rather plentiful, making them indeed useful and desirable.

\begin{proposition}
If $C \subset \mathbb{F}_{q^2}^\ell$ is Galois closed and $C=\langle \textbf{v} \rangle$ where $\textbf{v} = (v_1, \ldots, v_\ell)$, then $v_i = 0$ or $v_i^{q-1} = \beta$ for some fixed $\beta\in \mathbb{F}_{q^2}$ for all $1 \leq i \leq s$.  
\end{proposition}

\begin{proposition}
Let $C \subseteq \mathbb{F}_{q^2}^{\ell}$ be a $k$-dimensional code. Then $C$ is Galois closed if and only if $C = \bigoplus_i^k C_i$ where $C_i = C_i^q$ and $\dim(C_i) = 1$, i.e., $C$ is Galois closed if any only if can be decomposed as Galois closed codes of dimension $1$.
\end{proposition}

In general $C\neq C^q$, but they have the same parameters. According to Remark~\ref{codeC^q}, if $C$ is Galois closed, then its Hermitian and Euclidean duals are exactly the same, i.e. $C^{\perp_{H}}=C^{\perp_{E}}$. Furthermore, the following theorem shows that if $C$ is Galois closed, then its Euclidean dual is also Galois closed.

\begin{theorem}
A code $C$ over $\F_{q^2}$ is Galois closed if and only if its Euclidean dual is also Galois closed.
\end{theorem}

\begin{proof}
We only prove the forward direction as the reverse direction is nearly identical. Let $x\in C$ and $y\in C^{\perp_E}$. Then
$$\langle y^q, x\rangle_E = \sum_i y_i^q x_i = \bigg(\sum_i y_i^q x_i\bigg)^q = \sum_i y_i x_i^q = \langle y, x^q\rangle_E. $$

\noindent As $y\in C^{\perp_E}$, $y^q$ is in $(C^{\perp_E})^q$. Likewise, since $x\in C$ and $C = C^q$, we have that $\langle y, x^q\rangle_E = 0$ and so $\langle y^q, x\rangle_E = 0$. Thus $y^q \in C^{\perp_E}$ , i.e., $y\in (C^{\perp_E})^q$, and so $C^{\perp_E}\subseteq (C^{\perp_E})^q$. The reverse direction follows similarly. \qed
\end{proof}

\begin{proposition}\label{GaloisClosed}
Let $S=\left\{(c_1 , ..., c_n), \left(c_1 ^q , ..., c_n ^q \right), ..., \left(c_1 ^{q^{r-1} } , ..., c_n ^{q^{r-1}} \right)\right\} \subset \F_{q^r} ^n$, and $C=\langle S \rangle$ be the respective $[n, r\geq k]$-code. Then $C$ and its Euclidean dual are Galois closed. 
\end{proposition}

\begin{proof}
Without loss of generality, let $B=\{(c_1 , ..., c_n), (c_1 ^q , ..., c_n ^q ), ..., (c_1 ^{q^{k-1} } , ..., c_n ^{q^{k-1}} )\}=\{\mathbf{c_0} ,\mathbf{ c_1 }, ...,\mathbf{c_{k-1}}\}$ be a basis of $C$. Given $\mathbf{c_{\gamma}}=\sum_{i=0} ^{k-1} \gamma_i \mathbf{c_i} $, where $\gamma_i \in \mathbb{F}_{q^r}$ for $0 \leq i \leq k-1$, we have 
\begin{eqnarray}
\mathbf{c_{\gamma}}^q &=& (\sum_{i=0} ^{k-1} \gamma_i \mathbf{c_i} )^q = \sum_{i=0} ^{k-1} \gamma_i ^q \mathbf{c_{i+1}} = \sum_{i=0} ^{k-2} \gamma_i ^q \mathbf{c_{i+1}} + \gamma_{k-1}^q \sum_{i=0} ^{k-1} \beta_i \mathbf{c_{i}}, \mbox{ where } \beta_ i \in \F_{q^r}, i = 0,...,k-1;\nonumber \\
&=& \gamma_{k-1}^q \beta_0 \mathbf{c_0} + \sum_{i=1} ^{k-1} (\gamma_{i-1} ^q +\beta_i) \mathbf{c_{i}} \in C,
\end{eqnarray}
namely, $C^q \subset C$. 
Since $C$ and $C^q$ have the same dimension; therefore, $C$ is Galois closed.

Finally, it is enough to prove that $C^{\perp_E}\subset (C^q)^{\perp_E}$, since $(C^q)^{\perp_E} \mbox{ and } C^{\perp_E}$ have the same dimension. Let $\mathbf{x}=(x_1 ,...,x_n) \in C^{\perp_E}$. Therefore, 
\begin{equation}
\langle\mathbf{x} ^q ,\mathbf{c}_{\gamma}\rangle_E = \sum_{i=0} ^{k-1}  \langle \mathbf{x}^q, \gamma_i \mathbf{c_i} \rangle_E =\left( \sum_{i=0} ^{k-1}  \langle \mathbf{x}, \gamma_ i ^{q^{r-1}} \mathbf{c_{i-1}} \rangle_E \right)^q = 0
\end{equation}
where $\mathbf{c_{-1}}$ is written as linear combination of $\mathbf{c_{0}},...,\mathbf{c_{k-1}}$. \qed
\end{proof}

The Trace Representation~\cite[Theorem 5.1]{lingsole} of $\ell$-QC codes depends on the constituent codes and $q$-cyclotomic cosets. It turns out that the the constituent codes of a QC code being Galois closed directly determines whether $\C \subset \F_q ^{m\ell}$ is Galois closed.

\begin{theorem}\label{GC-iff-constituents-GC}
An $\ell$-QC code $\C \subset \F_{q^2}^{m\ell}$ is Galois closed if and only if its respective constituent codes are Galois closed.
\end{theorem}

\begin{proof}

Let $\mathcal{C}$ be an $\ell$-QC code with decomposition in~\eqref{lingsole-decomposition1}. Then
\begin{equation}
\mathcal{C}^q \cong \Bigg(\bigoplus_{j=1}^t ({(C^q)}{'}_j \oplus {{(C^q)}'_j}^{\perp_E})\Bigg) \oplus \bigg(\bigoplus_{i=1} ^s (C^q)_i \bigg).
\end{equation}
As the corresponding constituent codes of $\mathcal{C}^q$ and $\mathcal{C}$ are aligned, we have that $(C^q)'_j=C'^q_j$ and ${(C^q)'_j}^{\perp_E}=((C'_j)^{\perp_E})^q$ for all $1 \leq j \leq t$, and $(C^q)_i=C_i^q$ for each $1 \leq i \leq s$. Thus $\mathcal{C}=\mathcal{C}^q$ as codes if and only if $(C^q)_i=C_i^q=C_i$ for each $1 \leq i \leq s$, and $(C^q)'_j=C'^q_j=C'_j$ and ${(C^q)'_j}^{\perp_E}=({C'_j}^{\perp_E})^q={C'_j}^{\perp_E}$ for all $1 \leq j \leq t$. This can also be seen via the Trace Representation~\cite[Theorem 5.1]{lingsole} in conjunction with properties of the trace map. \qed
\end{proof}

We use these results to obtain the following family of $q^2-$ary QC codes that satisfy the square-root-like bound for their minimum distances.

\begin{theorem}\label{theorem:Galois_Hermitian}
Let $R_{q^2,m}^{\ell}$ be as in \eqref{eq:new_r_decomp}. Let $\mathcal{C}_{u}$ be the $m^u\ell $-length QC code whose constituent codes in ~\eqref{first_2t_general_level1} and~\eqref{last_s_general_level1} are all Galois closed and whose minimum distances satisfy~\eqref{chaincond}. Additionally, assume that the constituent codes in~\eqref{last_s_general_level1} are HSO. Then we have an infinite family $\{\mathcal{C}_u\}_{u\geq 1}$ of HSO $q^2$-ary QC Galois closed codes where $d(\C_u) \geq m^{u-1} d_{1,s}$ for each $u\geq 1$. Moreover,
\begin{itemize}
    \item[(i)] If $d(C_{1,t+1}) \ge c\sqrt{\ell}$ for some constant $c$, then all codes in $\{\C_u\}_{u\geq 2}$ satisfy the square-root-like bound for their minimum distances.
    \item[(ii)] Otherwise, there exist a positive integer $L >1 $ so that the all codes in the subfamily $\{\C_u\}_{u\geq L}$ of $\{\C_u\}_{u\geq 1}$ satisfies the square-root-like bound for their minimum distances.
\end{itemize}
\end{theorem}
\begin{remark}Analogous results to those of Theorem~\ref{theorem:Galois_Hermitian} also hold when the constituent codes in \eqref{last_s_general_level1} are HSD provided that $x^m-1$ has only $x-1$ as its self-reciprocal factor.
\end{remark}



\color{black}

\section{Constructions of QECCs from QC Codes}\label{QECC_constructions}

Given the construction of $\ell$-QC codes and results seen in~\cite{cem,lingsole}, along with the collection of quantum constructions in Subsection~\ref{quantumconstructions}, we present in this section a construction of quantum codes with good parameters.

\begin{example}\label{example_optimal}
Let $q=4$, $\ell=3$, $m=7$, and $\alpha$ be a primitive $7$-th root of unity. Let
\begin{eqnarray*}\label{CRT_decomposition_R^3}
R_{4,7}^3 &\cong& \left(\frac{\mathbb{F}_4 [x]}{\langle x^3 +x +1 \rangle} \right)^3 \oplus\left(\frac{\mathbb{F}_4 [x]}{\langle x^3 +x^2 +1 \rangle}\right)^3\oplus\left(\frac{\mathbb{F}_4 [x]}{\langle x -1 \rangle}\right)^3 \nonumber \\
&\cong & \mathbb{F}_4 ^3(\alpha) \oplus \mathbb{F}_4 ^3 (\alpha^3 ) \oplus \mathbb{F}_4 ^3.
\end{eqnarray*}

For any element $\gamma \in \F_4 (\alpha)\setminus \F_4$, let $C_{1,1} ' = \left\langle (\gamma,\gamma,\gamma)\right\rangle \subset \F_{4} ^3 (\alpha)$ , $C_{1,1}'' = (C_{1,1}')^{\perp_E} \subset \F_{4} ^3 (\alpha^3)$, and $C_{1,1} = \F_4 ^3$ be the constituent codes, whose parameters are $[3,1,3]_{64}$, $[3,2,2]_{64}$, and $[3,3,1]_4$, respectively. Notice that $C_{1,1} ', C_{1,1}''$ are Galois closed by Proposition~\ref{GaloisClosed} and $C_{1,1}$ is trivially HDC and Galois closed.

\color{black}
Let $\C \subset \F_4 ^{m\ell}$ be the respective $3$-QC code which, by Theorem~\ref{GC-iff-constituents-GC}, is also Galois-closed. By Lemma~\ref{dimensionlemma} and Theorem~\ref{lingsole}, $\C$ is a $[21,12,7]_4$-HDC (and also EDC) QC code, where the minimum distance was computed using \textsc{Magma}.  By Corollary~\ref{quantumcodefromdualcontaining}, it is possible to obtain a $[[21,3,d \geq 7]]_4$-quantum code $Q$ from the QC code $\mathcal{C}$.  According to~\cite{table}, the best-known possible minimum distance for a quaternary QECC with $n=21$ and $k=3$ is $7$.


\begin{center}
 \begin{tabular}{lccccc}
 \cline{2-6}
 \multicolumn{1}{l|}{}                                                                                                          & \multicolumn{1}{c|}{$[[16,8]]_4$} & \multicolumn{1}{c|}{$[[17,7]]_4$} & \multicolumn{1}{c|}{$[[18,6]]_4$} & \multicolumn{1}{c|}{$[[19,5]]_4$} & \multicolumn{1}{c|}{$[[20,4]]_4$} \\ \hline
 \multicolumn{1}{|l|}{Ex.~\ref{example_optimal} and Lemma~\ref{qeecc1} - item (2)} & \multicolumn{1}{c|}{$\geq 2$}   & \multicolumn{1}{c|}{$\geq 3$}   & \multicolumn{1}{c|}{$\geq 4$}   & \multicolumn{1}{c|}{$\geq 5$}   & \multicolumn{1}{c|}{$\geq 6$}   \\ \hline
 \multicolumn{1}{|l|}{Code Table~\cite{table}}                                                            & \multicolumn{1}{c|}{4-5}          & \multicolumn{1}{c|}{5-6}          & \multicolumn{1}{c|}{5-7}          & \multicolumn{1}{c|}{6-8}        & \multicolumn{1}{c|}{7-8}          \\ \hline
 & \multicolumn{1}{l}{}            & \multicolumn{1}{l}{}            & \multicolumn{1}{l}{}            & \multicolumn{1}{l}{}            & \multicolumn{1}{l}{}            \\ \cline{2-6} 
 \multicolumn{1}{l|}{}                                                                                                          & \multicolumn{1}{c|}{$[[22,3]]_4$} & \multicolumn{1}{c|}{$[[23,3]]_4$} & \multicolumn{1}{c|}{$[[24,3]]_4$} & \multicolumn{1}{c|}{$[[25,3]]_4$} & \multicolumn{1}{c|}{$[[26,3]]_4$} \\ \hline
 \multicolumn{1}{|l|}{Ex.~\ref{example_optimal} and Lemma~\ref{qeecc1} - item (1)} & \multicolumn{1}{c|}{$\geq 7$}   & \multicolumn{1}{c|}{$\geq 7$}   & \multicolumn{1}{c|}{$\geq 7$}   & \multicolumn{1}{c|}{$\geq 7$}   & \multicolumn{1}{c|}{$\geq 7$}   \\ \hline
 \multicolumn{1}{|l|}{Code Table~\cite{table}}                                                            & \multicolumn{1}{c|}{7-10}        & \multicolumn{1}{c|}{8-11}        & \multicolumn{1}{c|}{8-11}        & \multicolumn{1}{c|}{9-11}        & \multicolumn{1}{c|}{9-12}        \\ \hline
 \end{tabular}
\captionsetup{hypcap=false}
\captionof{table}{Parameters of good QECCs from our code construction \label{Table:More_good_QECCs}}
\end{center}

By Lemma~\ref{qeecc1}, it is still possible to construct more good QECCs from the code $\C$ which are listed in Table~\ref{Table:More_good_QECCs}. Although there is no significant improvement in the minimum distances for any of the above parameters, most of our codes achieve values that are within one or two of the best-known lower bounds. It is also worth noting that our example exhibits lower bounds for the minimum distance, so it may be possible to obtain codes with larger minimum distances.
\end{example}

\begin{example}
Still working with the decomposition of $R_{4,7} ^3$ provided in Example~\ref{example_optimal}, assume the same constituent codes $C_{1,1} '$ and $C_{1,1} ^{''}$ but let $C_{1,1} = \langle (0,1,\beta), (1,0,\beta^2)\rangle$ be the last constituent code, where $\beta \in \F_4$ so that $\beta^2 = \beta+1$. The code $C_{1,1}$ parameters $[3,2,2]_{4}$ and is EDC.

Let $\C \subset \F_4 ^{m\ell}$ be the respective $3$-QC code. By Corollary~\ref{cormainTh}, we have an infinity family $\{\C_u\}_{u\geq 1}$ of EDC QC codes whose their parameters are 
\begin{equation*}
 {\left[3\cdot 7^u , \frac{3\cdot 7^u +1}{2}, \geq 2\cdot 7^{u-1}\right]_4.}
\end{equation*}
From Corollary~\ref{quantumcodefromdualcontaining}, it is possible to obtain the infinity family of quantum codes $\{Q_u\}_{u\geq 1}$ with parameters \linebreak {$\left[\left[3\cdot 7^u, 1, \geq 2\cdot 7^{u-1}\right]\right]_4$.}
\end{example}
\color{black}

\section{Conclusion}
In this work, we introduced a recursive construction of infinite families of quasi-cyclic codes whose parameters are determined by their constituent codes. By carefully selecting these constituents, we obtained families that are self-orthogonal, dual-containing, or self-dual with respect to the Euclidean and Hermitian inner products. A key feature of our construction is that both the dimension and a lower bound for the minimum distance can be explicitly determined from the parameters of the initial constituent codes, allowing us to study how the parameters evolve throughout the family. In particular, we showed that the resulting codes satisfy a square-root-like lower bound on the minimum distance which ensures a nontrivial distance growth as the code length increases. Moreover, we can use the CSS construction to establish the existence of quantum error-correcting codes with good parameters that arise from our families. 

For future work, we aim to gain a deeper understanding of how using MDS codes as constituent codes affects the minimum distance of the respective QC Code. We also plan to explore sharper lower bounds for the minimum distance to improve the parameters of both the resulting QC code and its respective quantum code. Additionally, another plan is to relax the restriction in Section~\ref{section:dimension} that the constituent codes $C_{u, i}$ are given by $m^{u-1}$-copies of $C_{1,i}$. While this assumption simplified the analysis, allowing more general choices of $C_{u,i}$ may yield families with improved parameters. In each of these future directions, we aim to find the existence of quantum codes with good parameters.

\section{Acknowledgements}
The authors would like to thank the Institute for Computational and Experimental Research in Mathematics (ICERM) for providing funding and space for us to work in a mathematically rich environment during our week-long visit in Summer 2025. The first two authors would also like to thank Simons Laufer Mathematical Sciences Institute (SLMath) for providing funding for post-programmatic travel in 2024 following our Summer Research in Mathematics visit in 2023. The first author was supported by project FAPEMIG RED-00133-21. The second author was partially supported by the AMS-Simons Research Enhancement Grants for PUI 
Faculty while working on this project.  The three authors also thank Dr. Lara Ismert for her conversations that inspired us to use Galois closed codes to move freely between the Hermitian and Euclidean inner products and Dr. Markus Grassl for reading the first version of this manuscript and pointing out some inconsistencies.

\printcredits

\bibliographystyle{plain}

\bibliography{FFA_Submission}

\end{document}